\documentclass[final,3p,10pt,authoryear]{elsarticle}

\makeatletter
\def\ps@pprintTitle{%
 \let\@oddhead\@empty
 \let\@evenhead\@empty
 \def\@oddfoot{}%
 \let\@evenfoot\@oddfoot}
\makeatother

\usepackage{amssymb,amsthm,mathtools}
\usepackage[colorlinks]{hyperref}
\usepackage{natbib}

\usepackage{xcolor,dsfont,geometry}
\newtheorem{theorem}{Theorem}

\makeatletter \@addtoreset{equation}{section} \makeatother

\newcommand{\N}{\mathbb{N}}

\newcommand{\EE}{\mathbb{E}}

\newcommand{\bb}[1]{\boldsymbol{#1}}

\newcommand{\leqdef}{\vcentcolon=}

\newcommand{\blue}{\color{blue}}
\newcommand{\ind}{\mathds{1}}

\allowdisplaybreaks

\begin{document}

\begin{frontmatter}

    \title{Explicit formulas for the joint third and fourth central moments\\of the multinomial distribution}%

    \author[a1]{Fr\'ed\'eric Ouimet\texorpdfstring{\corref{cor1}\fnref{fn1}}{)}}%

    \address[a1]{California Institute of Technology, Pasadena, USA.}%

    \cortext[cor1]{Corresponding author}%
    \ead{ouimetfr@caltech.edu}%

    \begin{abstract}
        We give the first explicit formulas for the joint third and fourth central moments of the multinomial distribution, by differentiating the moment generating function. A general formula for the joint factorial moments was previously given in \cite{MR143299}.
    \end{abstract}

    \begin{keyword}
        multinomial distribution \sep simplex \sep central moments \sep third moment \sep fourth moment
        \MSC[2010]{Primary : 62E15 Secondary : 60E05}
    \end{keyword}

\end{frontmatter}

\vspace{-2mm}
\section{The multinomial distribution}\label{sec:intro}

    For any $d\in \N$, the $d$-dimensional (unit) simplex is defined by $\mathcal{S} \leqdef \big\{\bb{x}\in [0,1]^d : \sum_{i=1}^d x_i \leq 1\big\}$, and the probability mass function $\bb{k}\mapsto P_{\bb{k},m}(\bb{x})$ for $\bb{\xi} \leqdef (\xi_1,\xi_2,\dots,\xi_d) \sim \mathrm{Multinomial}\hspace{0.2mm}(m,\bb{x})$ is defined by
    \begin{equation}\label{eq:multinomial.probability}
        P_{\bb{k},m}(\bb{x}) \leqdef \frac{m!}{(m - \sum_{i=1}^d k_i)! \prod_{i=1}^d k_i!} \cdot (1 - \sum_{i=1}^d x_i)^{m - \sum_{i=1}^d k_i} \prod_{i=1}^d x_i^{k_i}, \quad \bb{k}\in \N_0^d \cap m\mathcal{S},
    \end{equation}
    where $m\in \N$ and $\bb{x}\in \mathcal{S}$.
    In this paper, our goal is to compute, for all $i,j,\ell,p\in \{1,2,\dots,d\}$,
    \begin{equation}
        \EE\big[(\xi_i - \EE[\xi_i])(\xi_j - \EE[\xi_j])(\xi_{\ell} - \EE[\xi_{\ell}])\big] \quad \text{and} \quad \EE\big[(\xi_i - \EE[\xi_i])(\xi_j - \EE[\xi_j])(\xi_{\ell} - \EE[\xi_{\ell}])(\xi_p - \EE[\xi_p])\big].
    \end{equation}

\section{Motivation}\label{sec:motivation}

    To the best of our knowledge, explicit formulas for the joint third and fourth central moments of the multinomial distribution have never been derived in the literature.
    These central moments can arise naturally, for example, when studying asymptotic properties, via Taylor expansions, of statistical estimators involving the multinomial distribution.
    For a given sequence of i.i.d.\ observations $\bb{X}_1,\bb{X}_2,\dots,\bb{X}_n$, two examples of such estimators are the Bernstein estimator for the cumulative distribution function
    \begin{equation}\label{eq:cdf.Bernstein.estimator}
        F_{n,m}^{\star}(\bb{x}) \leqdef \sum_{\bb{k}\in \N_0^d \cap m \mathcal{S}} \frac{1}{n} \sum_{i=1}^n \ind_{(-\bb{\infty},\frac{\bb{k}}{m}]}(\bb{X}_i) P_{\bb{k},m}(\bb{x}), \quad \bb{x}\in \mathcal{S}, ~m,n \in \N,
    \end{equation}
    and the Bernstein estimator for the density function (also called smoothed histogram)
    \begin{equation}\label{eq:histogram.estimator}
        \hat{f}_{n,m}(\bb{x}) \leqdef \sum_{\bb{k}\in \N_0^d \cap (m-1) \mathcal{S}} \frac{m^d}{n} \sum_{i=1}^n \ind_{(\frac{\bb{k}}{m}, \frac{\bb{k} + 1}{m}]}(\bb{X}_i) P_{\bb{k},m-1}(\bb{x}), \quad \bb{x}\in \mathcal{S}, ~m,n \in \N,
    \end{equation}
    over the $d$-dimensional simplex.
    Some of their asymptotic properties were investigated by \cite{MR0397977}, \cite{MR0638651}, \cite{MR0726014}, \cite{MR0791719}, \cite{MR0858109}, \cite{MR1437794}, \cite{MR1703623}, \cite{MR1712051}, \cite{MR1873330}, \cite{MR1881846}, \cite{MR1910059}, \cite{MR2068610}, \cite{MR2351744}, \cite{MR2395599}, \cite{MR2488150}, \cite{MR2662607}, \cite{MR2782409}, \cite{MR2960952}, \cite{MR2925964}, \cite{MR3174309}, \cite{MR3139345}, \cite{MR3412755}, \cite{MR3488598}, \cite{MR3740722}, \cite{MR3630225}, \cite{MR3983257} when $d = 1$, by \cite{MR1293514} when $d = 2$, and by \cite{arXiv:2002.07758,arXiv:2006.11756} for all $d\geq 1$, using a local limit theorem from \cite{arXiv:2001.08512} for the multinomial distribution (see also \cite{MR0478288}).
    The estimator \eqref{eq:histogram.estimator} is a discrete analogue of the Dirichlet kernel estimator introduced by \cite{doi:10.2307/2347365} and studied theoretically in \cite{MR1685301}, \cite{MR1718494}, \cite{MR1742101}, \cite{MR1985506} when $d = 1$ (among others), and in \cite{arXiv:2002.06956} for all $d\geq 1$.

\newpage
\section{Results}\label{sec:main.results}

    First, we compute the non-central moments of the multinomial distribution in \eqref{eq:multinomial.probability}.
    These moments were also calculated independently on a case-by-case basis in \cite{Newcomer2008phd,Newcomer_et_al_2008_tech_report}.

    \begin{theorem}[{\blue Non-central moments}]\label{thm:non.central.moments}
        Let $m^{(k)} \leqdef m (m-1) \dots (m-k+1)$ and let $\bb{\xi} \sim \mathrm{Multinomial}\hspace{0.3mm}(m,\bb{x})$.
        Then, for all $i,j,\ell,p\in \{1,2,\dots,d\}$,
        \begin{align}
            \EE[\xi_i]
            &= m \, x_i, \label{eq:lem:non.central.moments.eq.1} \\[2mm]
            \EE[\xi_i \xi_j]
            &= m^{(2)} x_i x_j + \ind_{\{i = j\}} m \, x_i, \label{eq:lem:non.central.moments.eq.2} \\[2mm]
            \EE[\xi_i \xi_j \xi_{\ell}]
            &= m^{(3)} x_i x_j x_{\ell} + \ind_{\{i = j = \ell\}} \, [3 \, m^{(2)} x_i^2 + m \, x_i] \label{eq:lem:non.central.moments.eq.3} \\[-0.5mm]
            &\quad+ \ind_{\{i = j \neq \ell \neq i\}} m^{(2)} x_i x_{\ell} + \ind_{\{i \neq j = \ell \neq i\}} m^{(2)} x_i x_j + \ind_{\{i \neq j \neq \ell = i\}} m^{(2)} x_j x_{\ell}, \notag \\[2mm]
            \EE[\xi_i \xi_j \xi_{\ell} \xi_p]
            &= m^{(4)} x_i x_j x_{\ell} x_p \label{eq:lem:non.central.moments.eq.4} \\[-1mm]
            &\quad+ \ind_{\{i = j \neq \ell \neq i\} \cap \{i \neq p \neq \ell\}} \, m^{(3)} x_i x_{\ell} x_p \notag \\
            &\quad+ \ind_{\{i \neq j = \ell \neq i\} \cap \{i \neq p \neq j\}} \, m^{(3)} x_i x_j x_p \notag \\
            &\quad+ \ind_{\{i \neq j \neq \ell = i\} \cap \{j \neq p \neq \ell\}} \, m^{(3)} x_j x_{\ell} x_p \notag \\[1mm]
            &\quad+ \ind_{\{i = j \neq \ell \neq i\} \cap \{i = p \neq \ell\}} \, [3 \, m^{(3)} x_i^2 x_{\ell} + m^{(2)} x_i x_{\ell}] \notag \\
            &\quad+ \ind_{\{i \neq j = \ell \neq i\} \cap \{i \neq p = j\}} \, [3 \, m^{(3)} x_i x_j^2 + m^{(2)} x_i x_j] \notag \\
            &\quad+ \ind_{\{i \neq j \neq \ell = i\} \cap \{j \neq p = \ell\}} \, [3 \, m^{(3)} x_j x_{\ell}^2 + m^{(2)} x_j x_{\ell}] \notag \\[1mm]
            &\quad+ \ind_{\{p = i\} \cap \{j \neq \ell \neq p \neq j\}} \, m^{(3)} x_i x_j x_{\ell} \notag \\
            &\quad+ \ind_{\{p = j\} \cap \{i \neq \ell \neq p \neq i\}} \, m^{(3)} x_i x_j x_{\ell} \notag \\
            &\quad+ \ind_{\{p = \ell\} \cap \{i \neq j \neq p \neq i\}} \, m^{(3)} x_i x_j x_{\ell} \notag \\[1mm]
            &\quad+ \ind_{\{i = j \neq \ell \neq i\} \cap \{i \neq p = \ell\}} \, [m^{(3)} x_i^2 x_{\ell} + m^{(3)} x_i x_{\ell}^2 + m^{(2)} x_i x_{\ell}] \notag \\
            &\quad+ \ind_{\{i \neq j = \ell \neq i\} \cap \{i = p \neq j\}} \, [m^{(3)} x_i x_j^2 + m^{(3)} x_i^2 x_j + m^{(2)} x_i x_j] \notag \\
            &\quad+ \ind_{\{i \neq j \neq \ell = i\} \cap \{j = p \neq \ell\}} \, [m^{(3)} x_j x_{\ell}^2 + m^{(3)} x_j^2 x_{\ell} + m^{(2)} x_j x_{\ell}] \notag \\[1mm]
            &\quad+ \ind_{\{i = j = \ell \neq p\}} \, [3 \, m^{(3)} x_i^2 x_p + m^{(2)} x_i x_p] \notag \\[1mm]
            &\quad+ \ind_{\{i = j = \ell = p\}} \, [6 \, m^{(3)} x_i^3 + 7 m^{(2)} x_i^2 + m x_i]. \notag
        \end{align}
    \end{theorem}

    \begin{proof}
        The moment generating function of $\bb{\xi}$ is
        \begin{equation}
            M_m(\bb{t}) = \bigg(1 - \sum_{i=1}^d x_i + \sum_{i=1}^d x_i \, e^{t_i}\bigg)^m.
        \end{equation}

        \vspace{-3mm}
        \noindent
        We have
        \begin{align}
            \frac{\partial}{\partial t_i} M_m(\bb{t})
            &= m M_{m-1}(\bb{t}) \, x_i \, e^{t_i}, \\[1mm]
            \frac{\partial^2}{\partial t_i \partial t_j} M_m(\bb{t})
            &= m^{(2)} M_{m-2}(\bb{t}) \, x_i \, e^{t_i} \, x_j \, e^{t_j} + \ind_{\{i = j\}} m M_{m-1}(\bb{t}) \, x_i \, e^{t_i}, \\[1mm]
            \frac{\partial^3}{\partial t_i \partial t_j \partial t_{\ell}} M_m(\bb{t})
            &= m^{(3)} M_{m-3}(\bb{t}) \, x_i \, e^{t_i} \, x_j \, e^{t_j} \, x_{\ell} \, e^{t_{\ell}} \\[-3mm]
            &\quad+ \ind_{\{i = j \neq \ell \neq i\}} m^{(2)} M_{m-2}(\bb{t}) \, x_i \, e^{t_i} \, x_{\ell} \, e^{t_{\ell}} \notag \\
            &\quad+ \ind_{\{i \neq j = \ell \neq i\}} m^{(2)} M_{m-2}(\bb{t}) \, x_i \, e^{t_i} \, x_j \, e^{t_j} \notag \\
            &\quad+ \ind_{\{i \neq j \neq \ell = i\}} m^{(2)} M_{m-2}(\bb{t}) \, x_j \, e^{t_j} \, x_{\ell} \, e^{t_{\ell}} \notag \\
            &\quad+ \ind_{\{i = j = \ell\}} \, [3 \, m^{(2)} M_{m-2}(\bb{t}) \, x_i^2 \, e^{2 t_i} + m M_{m-1}(\bb{t}) \, x_i \, e^{t_i}], \notag \\[1mm]
            \frac{\partial^4}{\partial t_i \partial t_j \partial t_{\ell} \partial t_p} M_m(\bb{t})
            &= m^{(4)} M_{m-4}(\bb{t}) \, x_i \, e^{t_i} \, x_j \, e^{t_j} \, x_{\ell} \, e^{t_{\ell}} \, x_p \, e^{t_p} \\[-2mm]
            &\quad+ \ind_{\{i = j \neq \ell \neq i\} \cap \{i \neq p \neq \ell\}} m^{(3)} M_{m-3}(\bb{t}) \, x_i \, e^{t_i} \, x_{\ell} \, e^{t_{\ell}} \, x_p \, e^{t_p} \notag \\
            &\quad+ \ind_{\{i \neq j = \ell \neq i\} \cap \{i \neq p \neq j\}} m^{(3)} M_{m-3}(\bb{t}) \, x_i \, e^{t_i} \, x_j e^{t_j} \, x_p \, e^{t_p} \notag \\
            &\quad+ \ind_{\{i \neq j \neq \ell = i\} \cap \{j \neq p \neq \ell\}} m^{(3)} M_{m-3}(\bb{t}) \, x_j \, e^{t_j} \, x_{\ell} \, e^{t_{\ell}} \, x_p \, e^{t_p} \notag \\[1mm]
            &\quad+ \ind_{\{i = j \neq \ell \neq i\} \cap \{i = p \neq \ell\}} \, \left[\hspace{-1mm}
                \begin{array}{l}
                    3 \, m^{(3)} M_{m-3}(\bb{t}) \, x_i^2 \, e^{2 t_i} \, x_{\ell} \, e^{t_{\ell}} \\
                    + m^{(2)} M_{m-2}(\bb{t}) \, x_i \, e^{t_i} \, x_{\ell} \, e^{t_{\ell}}
                \end{array}
                \hspace{-1mm}\right] \notag \\
            &\quad+ \ind_{\{i \neq j = \ell \neq i\} \cap \{i \neq p = j\}} \, \left[\hspace{-1mm}
                \begin{array}{l}
                    3 \, m^{(3)} M_{m-3}(\bb{t}) \, x_i \, e^{t_i} \, x_j^2 \, e^{2 t_j} \\
                    + m^{(2)} M_{m-2}(\bb{t}) \, x_i e^{t_i} \, x_j \, e^{t_j}
                \end{array}
                \hspace{-1mm}\right] \notag \\
            &\quad+ \ind_{\{i \neq j \neq \ell = i\} \cap \{j \neq p = \ell\}} \, \left[\hspace{-1mm}
                \begin{array}{l}
                    3 \, m^{(3)} M_{m-3}(\bb{t}) x_j \, e^{t_j} \, x_{\ell}^2 \, e^{2 t_{\ell}} \\
                    + m^{(2)} M_{m-2}(\bb{t}) \, x_j \, e^{t_j} \, x_{\ell} \, e^{t_{\ell}}
                \end{array}
                \hspace{-1mm}\right] \notag \\[1mm]
            &\quad+ \ind_{\{p = i\} \cap \{j \neq \ell \neq p \neq j\}} m^{(3)} M_{m-3}(\bb{t}) \, x_i \, e^{t_i} \, x_j \, e^{t_j} \, x_{\ell} \, e^{t_{\ell}} \notag \\
            &\quad+ \ind_{\{p = j\} \cap \{i \neq \ell \neq p \neq i\}} m^{(3)} M_{m-3}(\bb{t}) \, x_i \, e^{t_i} \, x_j \, e^{t_j} \, x_{\ell} \, e^{t_{\ell}} \notag \\
            &\quad+ \ind_{\{p = \ell\} \cap \{i \neq j \neq p \neq i\}} m^{(3)} M_{m-3}(\bb{t}) \, x_i \, e^{t_i} \, x_j \, e^{t_j} \, x_{\ell} \, e^{t_{\ell}} \notag \\[1mm]
            &\quad+ \ind_{\{i = j \neq \ell \neq i\} \cap \{i \neq p = \ell\}} \, \left[\hspace{-1mm}
                \begin{array}{l}
                    m^{(3)} M_{m-3}(\bb{t}) \, x_i^2 \, e^{2 t_i} \, x_{\ell} \, e^{t_{\ell}} \\
                    + m^{(3)} M_{m-3}(\bb{t}) \, x_i \, e^{t_i} \, x_{\ell}^2 \, e^{2 t_{\ell}} \\
                    + m^{(2)} M_{m-2}(\bb{t}) \, x_i \, e^{t_i} \, x_{\ell} \, e^{t_{\ell}}
                \end{array}
                \hspace{-1.5mm}\right] \notag \\
            &\quad+ \ind_{\{i \neq j = \ell \neq i\} \cap \{i = p \neq j\}} \, \left[\hspace{-1mm}
                \begin{array}{l}
                    m^{(3)} M_{m-3}(\bb{t}) \, x_i \, e^{t_i} \, x_j^2 \, e^{2 t_j} \\
                    + m^{(3)} M_{m-3}(\bb{t}) \, x_i^2 \, e^{2 t_i} \, x_j \, e^{t_j} \\
                    + m^{(2)} M_{m-2}(\bb{t}) \, x_i \, e^{t_i} \, x_j \, e^{t_j}
                \end{array}
                \hspace{-1.5mm}\right] \notag \\
            &\quad+ \ind_{\{i \neq j \neq \ell = i\} \cap \{j = p \neq \ell\}} \, \left[\hspace{-1mm}
                \begin{array}{l}
                    m^{(3)} M_{m-3}(\bb{t}) \, x_j \, e^{t_j} \, x_{\ell}^2 \, e^{2 t_{\ell}} \\
                    + m^{(3)} M_{m-3}(\bb{t}) \, x_j^2 \, e^{2 t_j} \, x_{\ell} \, e^{t_{\ell}} \\
                    + m^{(2)} M_{m-2}(\bb{t}) \, x_j \, e^{t_j} \, x_{\ell} \, e^{t_{\ell}}
                \end{array}
                \hspace{-1.5mm}\right] \notag \\[1mm]
            &\quad+ \ind_{\{i = j = \ell \neq p\}} \, \left[\hspace{-1mm}
                \begin{array}{l}
                    3 \, m^{(3)} M_{m-3}(\bb{t}) \, x_i^2 \, e^{2 t_i} \, x_p \, e^{t_p} \\
                    + m^{(2)} M_{m-2}(\bb{t}) \, x_i \, e^{t_i} \, x_p \, e^{t_p}
                \end{array}
                \hspace{-1mm}\right] \notag \\[1mm]
            &\quad+ \ind_{\{i = j = \ell = p\}} \, \left[\hspace{-1mm}
                \begin{array}{l}
                    6 m^{(3)} M_{m-3}(\bb{t}) \, x_i^3 \, e^{3 t_i} \\
                    + 7 m^{(2)} M_{m-2}(\bb{t}) \, x_i^2 \, e^{2 t_i} \\
                    + m M_{m-1}(\bb{t}) \, x_i \, e^{t_i}
                \end{array}
                \hspace{-1.5mm}\right]. \notag
        \end{align}
        By taking $\bb{t} = \bb{0}$, we get the conclusion.
    \end{proof}

    With some algebraic manipulations and a careful analysis, we can now obtain the central moments.

    \begin{theorem}[{\blue Central moments}]\label{thm:central.moments}
        Let $\bb{\xi} \sim \mathrm{Multinomial}\hspace{0.3mm}(m,\bb{x})$, then, for all $i,j,\ell,p\in \{1,2,\dots,d\}$,
        \begin{align}
            &\EE\big[\xi_i - \EE[\xi_i]\big] = 0, \label{eq:thm:central.moments.eq.1} \\[2mm]
            &\EE\big[(\xi_i - \EE[\xi_i])(\xi_j - \EE[\xi_j])\big] = m \, (x_i \ind_{\{i = j\}} - x_i x_j), \label{eq:thm:central.moments.eq.2} \\[2mm]
            &\EE\big[(\xi_i - \EE[\xi_i])(\xi_j - \EE[\xi_j])(\xi_{\ell} - \EE[\xi_{\ell}])\big] \notag \\[1mm]
            &\hspace{5mm}= m \, \big(2 x_i x_j x_{\ell} - \ind_{\{i = j\}} x_i x_{\ell} - \ind_{\{j = \ell\}} x_i x_j - \ind_{\{i = \ell\}} x_j x_{\ell} + \ind_{\{i = j = \ell\}} x_i\big), \label{eq:thm:central.moments.eq.3} \\[2mm]
            &\EE\big[(\xi_i - \EE[\xi_i])(\xi_j - \EE[\xi_j])(\xi_{\ell} - \EE[\xi_{\ell}])(\xi_p - \EE[\xi_p])\big] \notag \\[1mm]
            &\hspace{5mm}= (3 m^2 - 6 m) \, x_i x_j x_{\ell} x_p - (12 m^2 - 12 m) \, \ind_{\{i = j = \ell = p\}} x_i^3 \label{eq:thm:central.moments.eq.4} \\[0.5mm]
            &\hspace{5mm}\quad+ m^2 \left\{\hspace{-1mm}
                \begin{array}{l}
                    \ind_{\{i = j\}} x_i x_{\ell} x_p + \ind_{\{i = \ell\}} x_j x_{\ell} x_p + \ind_{\{i = p\}} x_i x_j x_{\ell} \\
                    \ind_{\{j = \ell\}} x_i x_j x_p + \ind_{\{j = p\}} x_i x_j x_{\ell} + \ind_{\{\ell = p\}} x_i x_j x_{\ell}
                \end{array}
                \hspace{-1mm}\right\} \notag \\[0.5mm]
            &\hspace{5mm}\quad+ (3 m^2 - 7 m) \, \ind_{\{i = j = \ell = p\}} x_i^2 + m \, \ind_{\{i = j = \ell = p\}} x_i \notag \\[0.5mm]
            &\hspace{5mm}\quad- (6 m^2 - 6 m) \left\{\hspace{-1mm}
                \begin{array}{l}
                    \ind_{\{i = j = \ell \neq p\}} x_i^2 x_p + \ind_{\{i = j = p \neq \ell\}} x_i^2 x_{\ell} \\
                    + \ind_{\{i = \ell = p \neq j\}} x_j x_{\ell}^2 + \ind_{\{j = \ell = p \neq i\}} x_i x_j^2
                \end{array}
                \hspace{-1mm}\right\} \notag \\[0.5mm]
            &\hspace{5mm}\quad- (2 m^2 - 2 m) \left\{\hspace{-1mm}
                \begin{array}{l}
                    \ind_{\{\ell \neq i = j \neq p\}} x_i x_{\ell} x_p + \ind_{\{j \neq i = \ell \neq p\}} x_j x_{\ell} x_p + \ind_{\{j \neq i = p \neq \ell\}} x_i x_j x_{\ell} \\
                    \ind_{\{i \neq j = \ell \neq p\}} x_i x_j x_p + \ind_{\{i \neq j = p \neq \ell\}} x_i x_j x_{\ell} + \ind_{\{i \neq \ell = p \neq j\}} x_i x_j x_{\ell}
                \end{array}
                \hspace{-1mm}\right\} \notag \\[0.5mm]
            &\hspace{5mm}\quad+ (m^2 - m) \, \big\{\ind_{\{i = j \neq \ell = p\}} x_i x_{\ell} + \ind_{\{i = p \neq j = \ell\}} x_i x_j + \ind_{\{i = \ell \neq j = p\}} x_j x_{\ell}\big\} \notag \\[0.5mm]
            &\hspace{5mm}\quad- m \, \big\{\ind_{\{i = j = \ell \neq p\}} x_i x_p + \ind_{\{i = j = p \neq \ell\}} x_i x_{\ell} + \ind_{\{i = \ell = p \neq j\}} x_j x_{\ell} + \ind_{\{j = \ell = p \neq i\}} x_i x_j\big\}. \notag
        \end{align}
    \end{theorem}

    \begin{proof}
        The expression \eqref{eq:thm:central.moments.eq.2} for the covariance follows directly from \eqref{eq:lem:non.central.moments.eq.1} and \eqref{eq:lem:non.central.moments.eq.2}.
        For the third central moments, \eqref{eq:lem:non.central.moments.eq.1}, \eqref{eq:lem:non.central.moments.eq.2} and \eqref{eq:lem:non.central.moments.eq.3} yield
        \begin{align*}
            &\EE\big[(\xi_i - \EE[\xi_i])(\xi_j - \EE[\xi_j])(\xi_{\ell} - \EE[\xi_{\ell}])\big] \\[1mm]
            &\hspace{5mm}= \EE[\xi_i \xi_j \xi_{\ell}] - m \, \EE[\xi_i \xi_j] \, x_{\ell} - m \, \EE[\xi_i \xi_{\ell}] \, x_j - m \, \EE[\xi_j \xi_{\ell}] \, x_i \\[0.5mm]
            &\hspace{5mm}\quad+ m^2 \EE[\xi_i] \, x_j x_{\ell} + m^2 \EE[\xi_j] \, x_i x_{\ell} + m^2 \EE[\xi_{\ell}] \, x_i x_j - m^3 x_i x_j x_{\ell} \\[0.5mm]
            &\hspace{5mm}= \left\{\hspace{-1mm}
                \begin{array}{l}
                    m^{(3)} \, x_i x_j x_{\ell} + \ind_{\{i = j = \ell\}} \, [3 \, m^{(2)} x_i^2 + m \, x_i] \\[0.5mm]
                    + \ind_{\{i = j \neq \ell \neq i\}} m^{(2)} x_i x_{\ell} + \ind_{\{i \neq j = \ell \neq i\}} m^{(2)} x_i x_j + \ind_{\{i \neq j \neq \ell = i\}} m^{(2)} x_j x_{\ell}
                \end{array}
                \hspace{-1mm}\right\} \\[0.5mm]
            &\hspace{5mm}\quad- m \, \big\{m^{(2)} x_i x_j + \ind_{\{i = j\}} m \, x_i\big\} x_{\ell} - m \, \big\{m^{(2)} x_i x_{\ell} + \ind_{\{i = {\ell}\}} m \, x_i\big\} \, x_j \\[0.5mm]
            &\hspace{5mm}\quad- m \, \big\{m^{(2)} x_j x_{\ell} + \ind_{\{j = \ell\}} m \, x_j\big\} x_i + 2 m^3 x_i x_j x_{\ell} \\[0.5mm]
            &\hspace{5mm}= 2 m \, x_i x_j x_{\ell} - m \, \ind_{\{i = j\}} x_i x_{\ell} - m \, \ind_{\{j = \ell\}} x_i x_j - m \, \ind_{\{i = \ell\}} x_j x_{\ell} + m \, \ind_{\{i = j = \ell\}} x_i \\[0.5mm]
            &\hspace{5mm}= m \, \big(2 x_i x_j x_{\ell} - \ind_{\{i = j\}} x_i x_{\ell} - \ind_{\{j = \ell\}} x_i x_j - \ind_{\{i = \ell\}} x_j x_{\ell} + \ind_{\{i = j = \ell\}} x_i\big).
        \end{align*}
        Similarly, by \eqref{eq:lem:non.central.moments.eq.1}, \eqref{eq:lem:non.central.moments.eq.2}, \eqref{eq:lem:non.central.moments.eq.3} and \eqref{eq:lem:non.central.moments.eq.4}, we have
        \begin{align*}
            &\EE\big[(\xi_i - \EE[\xi_i])(\xi_j - \EE[\xi_j])(\xi_{\ell} - \EE[\xi_{\ell}])(\xi_p - \EE[\xi_p])\big] \\[1mm]
            &\hspace{5mm}= \EE[\xi_i \xi_j \xi_{\ell} \xi_p] - m \, \EE[\xi_i \xi_j \xi_{\ell}] \, x_p - m \, \EE[\xi_i \xi_j \xi_p] \, x_{\ell} - m \, \EE[\xi_i \xi_{\ell} \xi_p] \, x_j - m \, \EE[\xi_j \xi_{\ell} \xi_p] \, x_i \\[0.75mm]
            &\hspace{5mm}\quad+ m^2 \EE[\xi_i \xi_j] \, x_{\ell} x_p + m^2 \EE[\xi_i \xi_{\ell}] \, x_j x_p + m^2 \EE[\xi_i \xi_p] \, x_j x_{\ell} \\[0.75mm]
            &\hspace{5mm}\quad+ m^2 \EE[\xi_j \xi_{\ell}] \, x_i x_p + m^2 \EE[\xi_j \xi_p] \, x_i x_{\ell} + m^2 \EE[\xi_{\ell} \xi_p] \, x_i x_j \\[0.75mm]
            &\hspace{5mm}\quad- m^3 \EE[\xi_i] \, x_i x_j x_{\ell} - m^3 \EE[\xi_j] \, x_i x_{\ell} x_p - m^3 \EE[\xi_{\ell}] \, x_i x_j x_p - m^3 \EE[\xi_p] \, x_i x_j x_{\ell} + m^4 x_i x_j x_{\ell} x_p \\[1mm]
            &\hspace{5mm}= \left\{\hspace{-1mm}
                \begin{array}{l}
                    m^{(4)} x_i x_j x_{\ell} x_p \notag \\
                    + \ind_{\{i = j \neq \ell \neq i\} \cap \{i \neq p \neq \ell\}} \, m^{(3)} x_i x_{\ell} x_p \notag \\
                    + \ind_{\{i \neq j = \ell \neq i\} \cap \{i \neq p \neq j\}} \, m^{(3)} x_i x_j x_p \notag \\
                    + \ind_{\{i \neq j \neq \ell = i\} \cap \{j \neq p \neq \ell\}} \, m^{(3)} x_j x_{\ell} x_p \notag \\[1.5mm]
                    + \ind_{\{i = j \neq \ell \neq i\} \cap \{i = p \neq \ell\}} \, [3 \, m^{(3)} x_i^2 x_{\ell} + m^{(2)} x_i x_{\ell}] \notag \\
                    + \ind_{\{i \neq j = \ell \neq i\} \cap \{i \neq p = j\}} \, [3 \, m^{(3)} x_i x_j^2 + m^{(2)} x_i x_j] \notag \\
                    + \ind_{\{i \neq j \neq \ell = i\} \cap \{j \neq p = \ell\}} \, [3 \, m^{(3)} x_j x_{\ell}^2 + m^{(2)} x_j x_{\ell}] \notag \\[1.5mm]
                    + \ind_{\{p = i\} \cap \{j \neq \ell \neq p \neq j\}} \, m^{(3)} x_i x_j x_{\ell} \notag \\
                    + \ind_{\{p = j\} \cap \{i \neq \ell \neq p \neq i\}} \, m^{(3)} x_i x_j x_{\ell} \notag \\
                    + \ind_{\{p = \ell\} \cap \{i \neq j \neq p \neq i\}} \, m^{(3)} x_i x_j x_{\ell} \notag \\[1.5mm]
                    + \ind_{\{i = j \neq \ell \neq i\} \cap \{i \neq p = \ell\}} \, [m^{(3)} x_i^2 x_{\ell} + m^{(3)} x_i x_{\ell}^2 + m^{(2)} x_i x_{\ell}] \notag \\
                    + \ind_{\{i \neq j = \ell \neq i\} \cap \{i = p \neq j\}} \, [m^{(3)} x_i x_j^2 + m^{(3)} x_i^2 x_j + m^{(2)} x_i x_j] \notag \\
                    + \ind_{\{i \neq j \neq \ell = i\} \cap \{j = p \neq \ell\}} \, [m^{(3)} x_j x_{\ell}^2 + m^{(3)} x_j^2 x_{\ell} + m^{(2)} x_j x_{\ell}] \notag \\[1.5mm]
                    + \ind_{\{i = j = \ell \neq p\}} \, [3 \, m^{(3)} x_i^2 x_p + m^{(2)} x_i x_p] \notag \\[1.5mm]
                    + \ind_{\{i = j = \ell = p\}} \, [6 m^{(3)} x_i^3 + 7 m^{(2)} x_i^2 + m x_i]
                \end{array}
                \hspace{-1mm}\right\} \\[1mm]
            &\hspace{5mm}\quad- m \left\{\hspace{-1mm}
                \begin{array}{l}
                    m^{(3)} \, x_i x_j x_{\ell} + \ind_{\{i = j = \ell\}} \, [3 \, m^{(2)} x_i^2 + m \, x_i] \\[0.5mm]
                    + \ind_{\{i = j \neq \ell \neq i\}} m^{(2)} x_i x_{\ell} + \ind_{\{i \neq j = \ell \neq i\}} m^{(2)} x_i x_j + \ind_{\{i \neq j \neq \ell = i\}} m^{(2)} x_j x_{\ell}
                \end{array}
                \hspace{-1mm}\right\} x_p \\[1mm]
            &\hspace{5mm}\quad- m \left\{\hspace{-1mm}
                \begin{array}{l}
                    m^{(3)} \, x_i x_j x_p + \ind_{\{i = j = p\}} \, [3 \, m^{(2)} x_i^2 + m \, x_i] \\[0.5mm]
                    + \ind_{\{i = j \neq p \neq i\}} m^{(2)} x_i x_p + \ind_{\{i \neq j = p \neq i\}} m^{(2)} x_i x_j + \ind_{\{i \neq j \neq p = i\}} m^{(2)} x_j x_p
                \end{array}
                \hspace{-1mm}\right\} x_{\ell} \\[1mm]
            &\hspace{5mm}\quad- m \left\{\hspace{-1mm}
                \begin{array}{l}
                    m^{(3)} \, x_i x_{\ell} x_p + \ind_{\{i = \ell = p\}} \, [3 \, m^{(2)} x_i^2 + m \, x_i] \\[0.5mm]
                    + \ind_{\{i = \ell \neq p \neq i\}} m^{(2)} x_i x_p + \ind_{\{i \neq \ell = p \neq i\}} m^{(2)} x_i x_{\ell} + \ind_{\{i \neq \ell \neq p = i\}} m^{(2)} x_{\ell} x_p
                \end{array}
                \hspace{-1mm}\right\} x_j \\[1mm]
            &\hspace{5mm}\quad- m \left\{\hspace{-1mm}
                \begin{array}{l}
                    m^{(3)} \, x_j x_{\ell} x_p + \ind_{\{j = \ell = p\}} \, [3 \, m^{(2)} x_j^2 + m \, x_j] \\[0.5mm]
                    + \ind_{\{j = \ell \neq p \neq j\}} m^{(2)} x_j x_p + \ind_{\{j \neq \ell = p \neq j\}} m^{(2)} x_j x_{\ell} + \ind_{\{j \neq \ell \neq p = j\}} m^{(2)} x_{\ell} x_p
                \end{array}
                \hspace{-1mm}\right\} x_i \\[1mm]
            &\hspace{5mm}\quad+ m^2 \big\{m^{(2)} x_i x_j + \ind_{\{i = j\}} m \, x_i\big\}  x_{\ell} x_p + m^2 \big\{m^{(2)} x_i x_{\ell} + \ind_{\{i = \ell\}} m \, x_i\big\} x_j x_p \\[1mm]
            &\hspace{5mm}\quad+ m^2 \big\{m^{(2)} x_i x_p + \ind_{\{i = p\}} m \, x_i\big\} x_j x_{\ell} + m^2 \big\{m^{(2)} x_j x_{\ell} + \ind_{\{j = \ell\}} m \, x_j\big\} x_i x_p \\[1mm]
            &\hspace{5mm}\quad+ m^2 \big\{m^{(2)} x_j x_p + \ind_{\{j = p\}} m \, x_j\big\} x_i x_{\ell} + m^2 \big\{m^{(2)} x_{\ell} x_p + \ind_{\{\ell = p\}} m \, x_{\ell}\big\} x_i x_j - 3 m^4 x_i x_j x_{\ell} x_p.
        \end{align*}
        We notice that all the terms with powers $m^3$ and $m^4$ cancel out with each other.
        The above is thus
        \begin{align*}
            &= \left\{\hspace{-1mm}
                \begin{array}{l}
                    (11 m^2 - 6 m) \, x_i x_j x_{\ell} x_p \notag \\[1mm]
                    + (- 3 m^2 + 2 m) \, \ind_{\{i = j \neq \ell \neq i\} \cap \{i \neq p \neq \ell\}} x_i x_{\ell} x_p \notag \\
                    + (- 3 m^2 + 2 m) \, \ind_{\{i \neq j = \ell \neq i\} \cap \{i \neq p \neq j\}} x_i x_j x_p \notag \\
                    + (- 3 m^2 + 2 m) \, \ind_{\{i \neq j \neq \ell = i\} \cap \{j \neq p \neq \ell\}} x_j x_{\ell} x_p \notag \\[1.5mm]
                    + (- 9 m^2 + 6 m) \, \ind_{\{i = j \neq \ell \neq i\} \cap \{i = p \neq \ell\}} \, x_i^2 x_{\ell} \notag \\
                    + (- 9 m^2 + 6 m) \, \ind_{\{i \neq j = \ell \neq i\} \cap \{i \neq p = j\}} \, x_i x_j^2 \notag \\
                    + (- 9 m^2 + 6 m) \, \ind_{\{i \neq j \neq \ell = i\} \cap \{j \neq p = \ell\}} \, x_j x_{\ell}^2 \notag \\[1.5mm]
                    + (1 m^2 - 1 m) \, \ind_{\{i = j \neq \ell \neq i\} \cap \{i = p \neq \ell\}} \, x_i x_{\ell} \notag \\
                    + (1 m^2 - 1 m) \, \ind_{\{i \neq j = \ell \neq i\} \cap \{i \neq p = j\}} \, x_i x_j \notag \\
                    + (1 m^2 - 1 m) \, \ind_{\{i \neq j \neq \ell = i\} \cap \{j \neq p = \ell\}} \, x_j x_{\ell} \notag \\[1.5mm]
                    + (- 3 m^2 + 2 m) \, \ind_{\{p = i\} \cap \{j \neq \ell \neq p \neq j\}} x_i x_j x_{\ell} \notag \\
                    + (- 3 m^2 + 2 m) \, \ind_{\{p = j\} \cap \{i \neq \ell \neq p \neq i\}} x_i x_j x_{\ell} \notag \\
                    + (- 3 m^2 + 2 m) \, \ind_{\{p = \ell\} \cap \{i \neq j \neq p \neq i\}} x_i x_j x_{\ell} \notag \\[1.5mm]
                    + (- 3 m^2 + 2 m) \, \ind_{\{i = j \neq \ell \neq i\} \cap \{i \neq p = \ell\}} \, x_i^2 x_{\ell} \notag \\
                    + (- 3 m^2 + 2 m) \, \ind_{\{i \neq j = \ell \neq i\} \cap \{i = p \neq j\}} \, x_i x_j^2 \notag \\
                    + (- 3 m^2 + 2 m) \, \ind_{\{i \neq j \neq \ell = i\} \cap \{j = p \neq \ell\}} \, x_j x_{\ell}^2 \notag \\[1.5mm]
                    + (- 3 m^2 + 2 m) \, \ind_{\{i = j \neq \ell \neq i\} \cap \{i \neq p = \ell\}} \, x_i x_{\ell}^2 \notag \\
                    + (- 3 m^2 + 2 m) \, \ind_{\{i \neq j = \ell \neq i\} \cap \{i = p \neq j\}} \, x_i^2 x_j \notag \\
                    + (- 3 m^2 + 2 m) \, \ind_{\{i \neq j \neq \ell = i\} \cap \{j = p \neq \ell\}} \, x_j^2 x_{\ell} \notag \\[1.5mm]
                    + (1 m^2 - 1 m) \, \ind_{\{i = j \neq \ell \neq i\} \cap \{i \neq p = \ell\}} \, x_i x_{\ell} \notag \\
                    + (1 m^2 - 1 m) \, \ind_{\{i \neq j = \ell \neq i\} \cap \{i = p \neq j\}} \, x_i x_j \notag \\
                    + (1 m^2 - 1 m) \, \ind_{\{i \neq j \neq \ell = i\} \cap \{j = p \neq \ell\}} \, x_j x_{\ell} \notag \\[1.5mm]
                    + (- 9 m^2 + 6 m) \, \ind_{\{i = j = \ell \neq p\}} \, x_i^2 x_p \notag \\[1.5mm]
                    + (1 m^2 - 1 m) \, \ind_{\{i = j = \ell \neq p\}} \, x_i x_p \notag \\[1.5mm]
                    + (- 18 m^2 + 12 m) \, \ind_{\{i = j = \ell = p\}} \, x_i^3 \notag \\[1.5mm]
                    + (7 m^2 - 7 m) \, \ind_{\{i = j = \ell = p\}} \, x_i^2 \notag \\[1mm]
                    + 1 m \, \ind_{\{i = j = \ell = p\}} x_i
                \end{array}
                \hspace{-1mm}\right\} \\[1mm]
            &\quad+ \left\{\hspace{-1mm}
                \begin{array}{l}
                    - 2 m^2 \, x_i x_j x_{\ell} x_p + 3 m^2 \, \ind_{\{i = j = \ell\}} \, x_i^2 x_p - 1 m^2 \, \ind_{\{i = j = \ell\}} \, x_i x_p \\[0.5mm]
                    + 1 m^2 \, \ind_{\{i = j \neq \ell \neq i\}} x_i x_{\ell} x_p + 1 m^2 \, \ind_{\{i \neq j = \ell \neq i\}} x_i x_j x_p + 1 m^2 \, \ind_{\{i \neq j \neq \ell = i\}} x_j x_{\ell} x_p
                \end{array}
                \hspace{-1mm}\right\} \\[1mm]
            &\quad+ \left\{\hspace{-1mm}
                \begin{array}{l}
                    - 2 m^2 \, x_i x_j x_{\ell} x_p + 3 m^2 \, \ind_{\{i = j = p\}} \, x_i^2 x_{\ell} - 1 m^2 \, \ind_{\{i = j = p\}} \, x_i x_{\ell} \\[0.5mm]
                    + 1 m^2 \, \ind_{\{i = j \neq p \neq i\}} x_i x_{\ell} x_p + 1 m^2 \, \ind_{\{i \neq j = p \neq i\}} x_i x_j x_{\ell} + 1 m^2 \, \ind_{\{i \neq j \neq p = i\}} x_j x_{\ell} x_p
                \end{array}
                \hspace{-1mm}\right\} \\[1mm]
            &\quad+ \left\{\hspace{-1mm}
                \begin{array}{l}
                    - 2 m^2 \, x_i x_j x_{\ell} x_p + 3 m^2 \, \ind_{\{i = \ell = p\}} \, x_i^2 x_j - 1 m^2 \, \ind_{\{i = \ell = p\}} \, x_i x_j \\[0.5mm]
                    + 1 m^2 \, \ind_{\{i = \ell \neq p \neq i\}} x_i x_j x_p + 1 m^2 \, \ind_{\{i \neq \ell = p \neq i\}} x_i x_j x_{\ell} + 1 m^2 \, \ind_{\{i \neq \ell \neq p = i\}} x_j x_{\ell} x_p
                \end{array}
                \hspace{-1mm}\right\} \\[1mm]
            &\quad+ \left\{\hspace{-1mm}
                \begin{array}{l}
                    - 2 m^2 \, x_i x_j x_{\ell} x_p + 3 m^2 \, \ind_{\{j = \ell = p\}} \, x_i x_j^2 - 1 m^2 \, \ind_{\{j = \ell = p\}} \, x_i x_j \\[0.5mm]
                    + 1 m^2 \, \ind_{\{j = \ell \neq p \neq j\}} x_i x_j x_p + 1 m^2 \, \ind_{\{j \neq \ell = p \neq j\}} x_i x_j x_{\ell} + 1 m^2 \, \ind_{\{j \neq \ell \neq p = j\}} x_i x_{\ell} x_p
                \end{array}
                \hspace{-1mm}\right\} \\[2mm]
            &= (3 m^2 - 6 m) \, x_i x_j x_{\ell} x_p - (12 m^2 - 12 m) \, \ind_{\{i = j = \ell = p\}} x_i^3  \\
            &\quad+ m^2 \left\{\hspace{-1mm}
                \begin{array}{l}
                    \ind_{\{i = j\}} x_i x_{\ell} x_p + \ind_{\{i = \ell\}} x_j x_{\ell} x_p + \ind_{\{i = p\}} x_i x_j x_{\ell} \\
                    \ind_{\{j = \ell\}} x_i x_j x_p + \ind_{\{j = p\}} x_i x_j x_{\ell} + \ind_{\{\ell = p\}} x_i x_j x_{\ell}
                \end{array}
                \hspace{-1mm}\right\} \\[0.5mm]
            &\quad+ (3 m^2 - 7 m) \, \ind_{\{i = j = \ell = p\}} x_i^2 + m \, \ind_{\{i = j = \ell = p\}} x_i \\[0.5mm]
            &\quad- (6 m^2 - 6 m) \left\{\hspace{-1mm}
                \begin{array}{l}
                    \ind_{\{i = j = \ell \neq p\}} x_i^2 x_p + \ind_{\{i = j = p \neq \ell\}} x_i^2 x_{\ell} \\
                    + \ind_{\{i = \ell = p \neq j\}} x_j x_{\ell}^2 + \ind_{\{j = \ell = p \neq i\}} x_i x_j^2
                \end{array}
                \hspace{-1mm}\right\} \notag \\[0.5mm]
            &\quad- (2 m^2 - 2 m) \left\{\hspace{-1mm}
                \begin{array}{l}
                    \ind_{\{\ell \neq i = j \neq p\}} x_i x_{\ell} x_p + \ind_{\{j \neq i = \ell \neq p\}} x_j x_{\ell} x_p + \ind_{\{j \neq i = p \neq \ell\}} x_i x_j x_{\ell} \\
                    \ind_{\{i \neq j = \ell \neq p\}} x_i x_j x_p + \ind_{\{i \neq j = p \neq \ell\}} x_i x_j x_{\ell} + \ind_{\{i \neq \ell = p \neq j\}} x_i x_j x_{\ell}
                \end{array}
                \hspace{-1mm}\right\} \\[0.5mm]
            &\quad+ (m^2 - m) \, \big\{\ind_{\{i = j \neq \ell = p\}} x_i x_{\ell} + \ind_{\{i = p \neq j = \ell\}} x_i x_j + \ind_{\{i = \ell \neq j = p\}} x_j x_{\ell}\big\} \\[0.5mm]
            &\quad- m \, \big\{\ind_{\{i = j = \ell \neq p\}} x_i x_p + \ind_{\{i = j = p \neq \ell\}} x_i x_{\ell} + \ind_{\{i = \ell = p \neq j\}} x_j x_{\ell} + \ind_{\{j = \ell = p \neq i\}} x_i x_j\big\}.
        \end{align*}
        This ends the proof.
    \end{proof}

\section{Conclusion}

    In this short paper, we found explicit expressions for the third and fourth central and non-central moments of the multinomial distribution.
    The joint factorial moments were previously calculated in the literature (by \cite{MR143299}), but not the third and fourth central moments.

\section*{Abbreviations}

    i.i.d.\ : independent and identically distributed

\section*{Declarations}

    \subsection{Availability of data and material}

        Not applicable.

    \subsection{Funding}

        The author is supported by a postdoctoral fellowship from the NSERC (PDF) and a supplement from the FRQNT (B3X).

    \subsection{Competing interests}

        The author declares no conflict of interest.

    \subsection{Author's contributions}

        All contributions were made by the sole author of the article, Fr\'ed\'eric Ouimet.

    \subsection{Acknowledgments}

        Not applicable.

%
%

\nocite{MR3825458}
\nocite{Ouimet2019phd}

\bibliographystyle{authordate1}
\bibliography{Ouimet_2020_third_fourth_moments_multinomial_bib}

\end{document}